\documentclass[11pt]{article}
\usepackage[paper=a4paper,dvips,top=1.5cm,left=1.5cm,right=1.5cm,bottom=1.5cm]{geometry}
\usepackage{amsfonts}
\usepackage{stmaryrd}
\usepackage{amsmath}
\usepackage{cases}
\usepackage{amssymb}
\usepackage{graphicx}
\usepackage{latexsym}
\usepackage{color}
\usepackage{cite}
\usepackage{comment}
\usepackage{makeidx}
\usepackage{times}
\usepackage{graphicx}
\usepackage{fancybox}

\usepackage{algorithmic}
\usepackage{algorithm}

\def\reals{{\mathbb R}}

\newcommand{\frechet}{Fr\'echet}

\newcommand{\dfd}{d_{dF}}

\begin{document}
\newtheorem{theorem}{Theorem}[section]
\newtheorem{lemma}[theorem]{Lemma}
\newenvironment{proof}[1][Proof]{\begin{trivlist}
\item[\hskip \labelsep {\bfseries #1}]}{\end{trivlist}}
\newenvironment{definition}[1][Definition]{\begin{trivlist}
\item[\hskip \labelsep {\bfseries #1}]}{\end{trivlist}}
\newenvironment{example}[1][Example]{\begin{trivlist}
\item[\hskip \labelsep {\bfseries #1}]}{\end{trivlist}}
\newenvironment{remark}[1][Remark]{\begin{trivlist}
\item[\hskip \labelsep {\bfseries #1}]}{\end{trivlist}}

\newcommand{\qed}{\nobreak \ifvmode \relax \else
      \ifdim\lastskip<1.5em \hskip-\lastskip
      \hskip1.5em plus0em minus0.5em \fi \nobreak
      \vrule height0.75em width0.5em depth0.25em\fi}

\title{Complexity and Algorithms for the Discrete Fr\'{e}chet Distance Upper Bound with Imprecise Input}

\author{Chenglin Fan, Binhai Zhu \\
 Department of Computer Science\\
 Montana State
University,
Bozeman, MT~59717, USA\\
chenglin.fan@msu.montana.edu, bhz@cs.montana.edu}

\maketitle
\begin{abstract}
We study the problem of computing the upper bound of the discrete
Fr\'{e}chet distance for imprecise input, and prove that the problem
is NP-hard. This solves an open problem posed in 2010 by Ahn
\emph{et al}. If shortcuts are allowed, we show that the upper bound
of the discrete Fr\'{e}chet distance with shortcuts for imprecise
input can be computed in polynomial time and we present several
efficient algorithms.
\end{abstract}
\section{Introduction}


The Fr\'{e}chet distance is a natural measure of similarity between
two curves \cite{AltG95}. The Fr\'{e}chet distance between two
curves is often referred to as the ``dog-leash distance''. Imagine a
dog and its handler are walking on their respective curves,
connected by a leash, and they both can control their speed but
cannot walk back. The Fr\'{e}chet distance of these two curves is
the minimum length of any leash necessary for the dog and the
handler to move from their starting points on the two curves to
their respective endpoints. Alt and Godau~\cite{AltG95} presented an
algorithm to compute the Fr\'{e}chet distance between two polygonal
curves of $n$ and $m$ vertices in $O(nm\log(nm))$ time. There has
been a lot of applications using the Fr\'{e}chet distance to do
pattern/curve matching. For instance, Fr\'{e}chet distance has been
extended to
 graphs (maps)~\cite{AltERW03,BPSW05}, to piecewise smooth
curves~\cite{Rote07}, to simple polygons~\cite{BBW08}, to
surfaces~\cite{AltB05}, to network distance~\cite{chenglin2011}, and
to the case when there is a speed limit \cite{Maheshwari2011}, etc.

On the other hand, Fr\'{e}chet distance is sensitive to local
errors, a small local change  could change the Fr\'{e}chet distance
greatly. In order to handle this kind of outliers, Driemel and
Har-Peled \cite{Driemel2013} introduced the Fr\'{e}chet distance
with shortcuts.

A slightly simpler version of the Fr\'{e}chet distance is the
\emph{discrete Fr\'{e}chet distance}, where only the vertices of
polygonal curves are considered. In terms of using a symmetric
example, we could imagine that two frogs, connected by a thread, hop
on two polygonal chains and each can hop from a vertex to the next
or wait, but can never hop back. Then, the discrete \frechet\
distance is the minimum length thread for the two frogs to reach the
ends of their respective chains. When we add a lot of points
(vertices) evenly on two polygonal chains, the discrete \frechet\
distance gives a natural approximation for the (continuous)
\frechet\ distance. The discrete \frechet\ distance is more suitable
for some applications, like protein structure alignment
\cite{JiangXZ08,Zhu07}, in which case each vertex represents the
$\alpha$-carbon atom of an amino acid. In this case, using the
(continuous) \frechet\ distance would produce some result which is
not biologically meaningful. In this paper, we focus  on the
discrete \frechet\ distance.

It takes $O(mn)$ time to compute the discrete \frechet\ distance
using a standard dynamic programming
technique~\cite{Wien94computingdiscrete}. Recently, this bound was
slightly improved \cite{AgarwalAKS14}. Most of the important
applications regarding the discrete \frechet\ distance are
biology-related \cite{JiangXZ08,Zhu07,Wylie2013,Chenglin15}. Some of
the other applications using the discrete \frechet\ distance just
study the corresponding problem using the (continuous) \frechet\
distance. For instance, given a polygonal curve $P$ and set of
points $S$, Maheshwari {\em et al.} studied the problem of computing
a polygonal curve through $S$ which has a minimum Fr\'{e}chet
distance to $P$ \cite{MaheshwariSSZ11}. The corresponding problem
using the discrete \frechet\ distance is studied in
\cite{Wylie2014}.

It is worth mentioning that, symmetric to the \frechet\ distance
with shortcuts \cite{Driemel2013}, the discrete \frechet\ distance
with shortcuts was also studied by Avraham {\em et al.} recently
\cite{Avraham2014}. A novel technique, based on distance selection,
was designed to compute the discrete \frechet\ distance with
shortcuts efficiently. In Section 4, we will also use the discrete
\frechet\ distance with shortcuts to compute the corresponding upper
bounds for imprecise input.

The  computational geometry with imprecise objects has drawn much
interest to researchers since a few years ago. There are two models:
one is the {\em continuous} model, where a precise point is selected
from an erroneous region (say a disk, or rectangle)
\cite{loffler2006}; the other is the {\em discrete} or {\em
color-spanning} model, where a precise point is selected from
several discrete objects with the same color and all colors must be
selected \cite{Abellanas2001}. We will mainly focus on the
continuous model, but will also touch the color-spanning model.
(There is another {\em probabilistic} model, which is not relevant
to this paper. Hence, we will skip that one.) A lot of algorithms
have been designed to handle imprecise geometric problems on both
models. For the continuous model, there are algorithms to handle
imprecise data for computing the Hausdorff
distance~\cite{loffler2009}, Voronoi diagram~\cite{Evans2008},
planar convex hulls~\cite{loffler2006,wenqi2012} and Delaunay
triangulations~\cite{loffler2010,Khanban2003}.

Ahn {\em et al}. studied the  problem of computing the discrete
Fr\'{e}chet distance between two imprecise point sequences, and gave
an efficient algorithm for computing the lower bound (of the
distance) and efficient approximation algorithms for the
corresponding upper bound (under a realistic assumption)
\cite{AhnIsaac2010,Ahn2012}. It is unknown whether computing the
discrete \frechet\ distance upper bound for imprecise input is
polynomially solvable or not, so Ahn {\em et al.} left that as an
open problem \cite{AhnIsaac2010,Ahn2012}. In this paper, we proved
that the problem is in fact NP-hard. We also consider the same
problem under the discrete \frechet\ distance with shortcuts and
give efficient polynomial-time solutions.

The paper is organized as follows. In Section 2, we give the
necessary definitions. In Section 3, we prove that the discrete
\frechet\ distance upper bound for imprecise input is NP-hard, which
is separated into several subsections due to the difficulty. In
Section 4, we consider the problem of computing the discrete
\frechet\ distance upper bound for imprecise input. In Section 5, we
conclude the paper.

\section{Preliminaries}

Throughout this paper, we use $d(a,b)$ for the Euclidean distance
between points $a$ and $b$, possibly in $\reals^k$, where $k$ is any
positive integer.

We first define the discrete \frechet\ distance as follows. Let
$A=(a_1\ldots,a_n)$ and $B=(b_1,\ldots,b_m)$ be two sequences of points
of size $n$ and $m$ respectively, in $\reals^k$. The discrete
\frechet\ distance $\dfd(A,B)$ between $A$ and $B$ is defined using
the following graph. Given a distance $\delta > 0$ and consider the
Cartesian product $A \times B$ as the vertex set of a directed graph
$G_\delta$ whose edge set is
\begin{align*}
    E_\delta = & \big\{ \big((a_i, b_j), (a_{i+1}, b_j)\big) \; | \; d(a_i, b_j), d(a_{i+1}, b_j) \le \delta \big\} \; \cup\\
    & \big\{ \big((a_i, b_j), (a_i, b_{j+1})\big) \; | \; d(a_i, b_j), d(a_i, b_{j+1}) \le \delta \big\} \; \cup\\
    & \big\{ \big((a_i, b_j), (a_{i+1}, b_{j+1})\big) \; | \; d(a_i, b_j), d(a_{i+1}, b_{j+1}) \le \delta \big\}\, .
\end{align*}
Then, $\dfd(A,B)$ is the smallest $\delta > 0$ for which $(a_n,
b_m)$ can be reached from $(a_1, b_1)$ in the graph $G_\delta$.

\begin{definition} For a region
$q_i$, a precise point $a_i$ is called a realization of $q_i$ if
$a_i\in q_i$; For a region sequences $Q=(q_1,q_2,...,q_n)$, the
precise point sequence $A=(a_1,a_2,...,a_n)$ is called a
\emph{realization} of $Q$ if we have $a_i \in q_i$ for all $1\leq
i\leq n$.
\end{definition}

For the discrete \frechet\ distance of imprecise input, we use the
same notions such that the \emph{realization} of an imprecise input
sequence as in \cite{AhnIsaac2010}. To be consistent with these
notations, we also use $F(A,B)$ to denote the discrete Fr\'{e}chet
distance between $A$ and $B$ (i.e., $F(A,B)=\dfd(A,B)$).

\begin{definition} 
For two region sequences $Q=(q_1,q_2,...,q_n)$ and
$H=(h_1,h_2,...,h_m)$, $A=(a_1,a_2,...,a_n)$ (resp.
$B=(b_1,b_2,...,b_m)$) is a possible realization of $H$ (resp. $Q$)
if we have $a_i \in q_i,b_j\in h_j$ for all $1\leq i\leq n, 1\leq
j\leq m$. The Fr\'{e}chet distance upper bound $F^{\max}(Q,H)=\max
\{F(A,B)\}$, where $A$ (resp. $B$) is a possible realization of $Q$
(resp. $H$).
\end{definition}

We comment that for region (or imprecise vertex) sequences, to
obtain decent algorithmic bounds, we mainly focus on the regions as
balls (disks in 2d) in Section 4. (Though with some extra twist, it
might be possible to handle square or rectangular regions as well.)
But in the proof of NP-hardness, the imprecise regions are
rectangles in Section 3.

We show in the next section that computing $F^{\max}(Q,H)$ is
NP-hard, which was an open problem posed by Ahn {\em et al.} in
\cite{AhnIsaac2010}.

\section{Computing the discrete Fr\'{e}chet distance upper bound of imprecise input is NP-hard}

In this section, we prove that deciding $F^{\max}(Q,H)\leq \epsilon$
is NP-hard. In fact, this holds even when $H$ is a precise vertex
sequence, and $Q$ is an imprecise vertex sequence (where each vertex
is modeled as a rectangle, not necessarily axis aligned). As the proof is quite
complex, we separate it in several parts.

\subsection{NP-hardness of an induced  subgraph connectivity problem of colored sets}

Firstly, we prove that another induced  subgraph connectivity
problem of colored sets is NP-hard, which is useful for the proof of
deciding $F^{\max}(Q,H)\leq \epsilon$.  We define the induced
subgraph connectivity problem of colored sets (ISCPCS) as follows:
let $G$ be the graph with $n$ vertices  and each vertex is colored
by one of the $m$ colors in the plane, a fixed source vertex $s$, a
fixed destination vertex $t$, and some directed edges between the
vertices (where no two edges cross), choose an induced subgraph
$G_s$ consisting of exactly one vertex of each color such that in
$G_s$ there is no path from $s$ to $t$. For an example, see
Figure~\ref{fig:Np_proof}. We prove that the ISCPCS problem is
NP-hard by a reduction from 3SAT.

\begin{figure}
 \centering
 \includegraphics[width=0.6\textwidth]{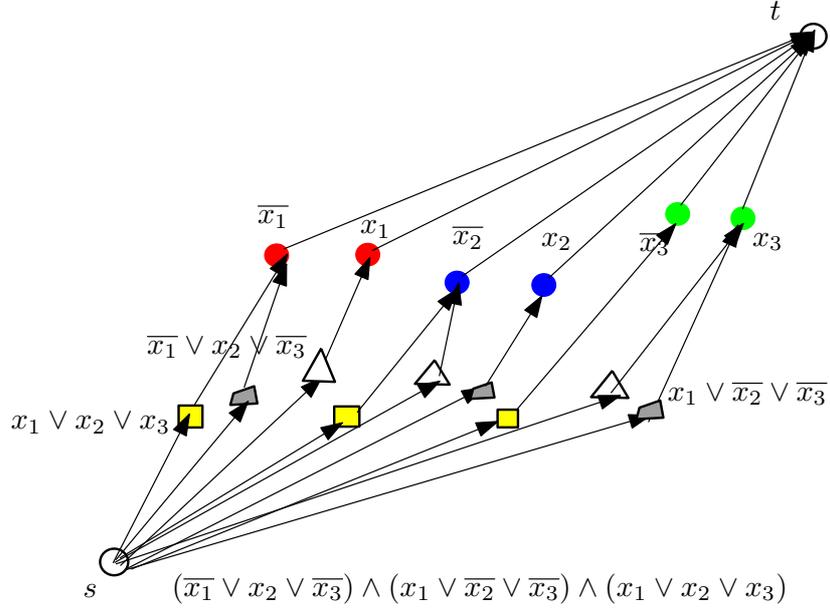}
 \caption{Illustration of the constructed directed colored graph $G$ from 3SAT. We also use different shapes for different clause vertices.}
  \label{fig:Np_proof}
\end{figure}

\begin{lemma}
ISCPCS is NP-hard.
\end{lemma}

The detailed proof is in the appendix, an example is given in
Figure~\ref{fig:Np_proof}.

\subsection{The free space diagram}

\begin{figure}
 \centering
 \includegraphics{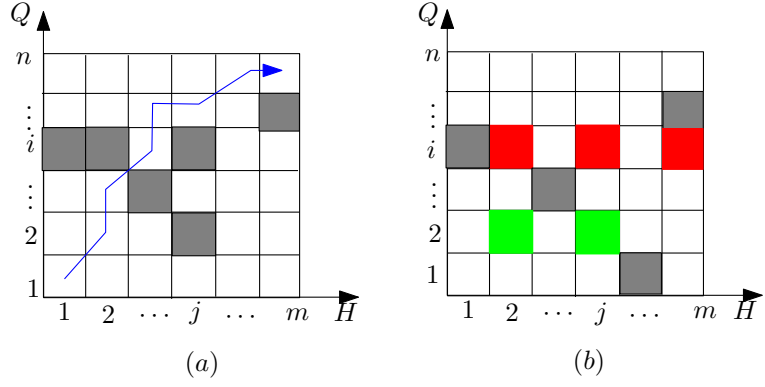}
 \caption{Illustration of the free space diagram of discrete Fr\'{e}chet distance with precise input (a);
  and, the free space diagram of discrete Fr\'{e}chet distance with imprecise input (b).}
  \label{fig:de_space}
\end{figure}

The free space diagram of the  discrete Fr\'{e}chet distance between
a realization of $Q,H$ is composed of a grid of $n\times m$ cells,
where $n$ and $m$ are the number of vertices in $Q$ and $H$
respectively. We first consider the case when both $Q,H$ are
precise. In this case, let $q_i$ and $h_j$ denote the $i$-th and
$j$-th vertex of $Q,H$ respectively. Each pair $(q_i,h_j)$
corresponds to the cell in the $i$-th row and the $j$-th column.
From the definition of the discrete \frechet\ distance, it
corresponds to a monotone path in the grid from cell (1,1) to
$(n,m)$. In the sequel, for the ease of description, we sometimes
loosely call such a path ``a monotone path". We cover the details
regarding such a path next.

Cell $C[i,j]=(q_i,h_j)$ is painted {\em white} if $d(q_i,h_j)\leq
\epsilon$,
which indicates that this cell can be passed by a potential monotone
path. Cell $C[i,j]=(q_i,h_j)$ is painted {\em gray} if $d(q_i,h_j)>
\epsilon$, which indicates that this grid cannot be passed by any
monotone path. Each cell $C[i,j]$ could reach its monotone
neighboring cell $C[i,j+1],C[i+1,j]$ or $C[i+1,j+1]$ if both of them
are painted white. The discrete \frechet\ distance is the minimum
$\epsilon$ such that there is a path from cell (1,1) to $(n,m)$ and
the path is monotone in both horizontal and vertical directions.
See Figure~\ref{fig:de_space} (a) for an example.

Now, we consider the free space diagram when $H=(h_1,h_2,...,h_{m})$
is a precise vertex sequence and $Q=(q_1,q_2,...,q_{n})$ is an
imprecise region sequence. There are several cases below.

\begin{itemize}
\item (1) If $d(q,h_j)\leq \epsilon, \forall q\in q_i $, then the cell
$C[i,j]$ is painted white and could be passed.
\item (2) If $d(q,h_j)> \epsilon, \forall q\in q_i $, then the cell $C[i,j]$
is painted gray and cannot be passed.
\item (3-a) There are two vertices $h_i, h_j$ and an imprecise vertex
$q_k$ satisfying either $d(q,h_i)\leq \epsilon$ or $d(q,h_j)\leq
\epsilon, \forall q\in q_k$, see
Figure~\ref{fig:spread_variable}(a). Then, we paint the cell
$C[k,i],C[k,j]$ with the same color, which show that either $C[k,i]$
or $C[k,j]$ can be passed, see Figure~\ref{fig:de_space} (b). This
case will be designed as a variable gadget.
\item (3-b) There are three vertices $h_i,h_j,h_k$ and an imprecise vertex
$q_x$ satisfying\\
$d(q,h_i)\leq \epsilon $, or $d(q,h_j)\leq \epsilon$ or
$d(q,h_k)\leq \epsilon, \forall q\in q_x$, see
Figure~\ref{fig:spread_variable}(b). Then we paint the cell
$C[x,i],C[x,j]$ and $C[x,k]$ with the same color.  This case can be
designed as a clause gadget. Of course, it is possible that more
than one of the cells $C[x,i],C[x,j],C[x,k]$ might be passed at the
same time. But our objective is to make the discrete Fr\'{e}chet
distance as large as possible when only one of them is passed.
\end{itemize}

In fact, there could be more complicated cases than the three cases
above, but we do not need them in our construction.

\begin{figure}
 \centering
 \includegraphics{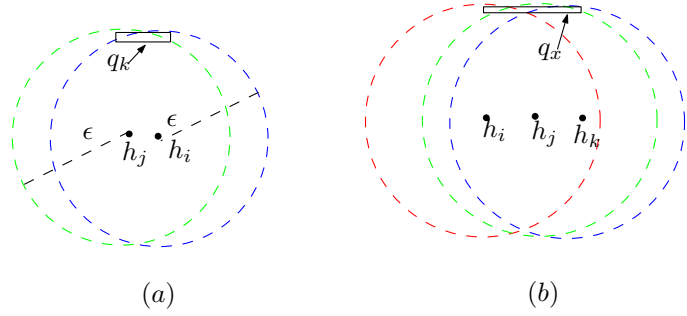}
 \caption{Illustration of the variable gadget and clause gadget.}
  \label{fig:spread_variable}
\end{figure}


\subsection{The grid graph for the color-spanning set}

The free space diagram of the discrete Fr\'{e}chet distance is
really a directed grid graph. Now we show how to convert the ISCPCS
instance, e.g., in Figure~\ref{fig:Np_proof}, into a grid graph. The
basic steps are as follows: the grid has $n+m+3$ rows and $3m+2n$
columns, and the colored cells in the grid correspond to the colored
vertices in ISCPCS. The details are step by step as follows.

\begin{enumerate}
\item For the first row (from bottom up), all the cells are
painted white, which means that all cells can be passed. The
motivation is to make the starting cell (the lower-left cell
$C[1,1]$, which corresponds to the start node $s$ in ISCPCS) in the
grid graph reachable to all the colored clause cells (which
correspond to clause vertices in ISCPCS).

\item From the 2nd row to the $(m+1)$-th row ($m$ is the number of clauses
in the 3SAT instance from which the ISCPCS instance is constructed),
each row has three cells with the same color. We call them {\em
clause cells}, corresponding to the three clause vertices of the
same color in ISCPCS. Each column has at most one clause cell. If
there is a clause cell $C(k,j)$ in the $j$-th column, then the cells
$C[i,j], 2\leq i\leq n+1, i\neq k$, are painted white and could be
passed. If there is no clause cell in the $j$-th column, then cells
$C[i,j], 2\leq i\leq n+1$, are painted by gray and could not be
passed.

\item We do not put any clause cell in the $(m+2)$-th row. If there exists
a clause cell $C[k,j]$, $2\leq k\leq n+1$, in the $j$-th column,
then $C[m+2,j]$ is painted white and can be passed; otherwise, the
cell $C[m+2,j]$ is painted gray and could not be passed.

\item From the $(m+3)$-th row to the $(n+m+2)$-th row ($n$ is the number of variables in the 3SAT instance from which the ISCPCS instance is
constructed), each row has two cells with the same color. We call
them {\em variable cells}, which correspond to two variable vertices
in the ISCPCS instance. (For an example, see the cells painted with
number $4$ in Figure~\ref{fig:spread_fff}.)  Each column has at most
one variable cell. If there is a variable cell $C[k,j]$ in the
$j$-th column, then the cells $C[i,j], m+3\leq i\leq n+m+2, i\neq
k$, are painted white and could be passed. If there is no variable
cell in the $j$-th column, then cells $C[i,j], m+3\leq i\leq n+m+2,
i\neq k$, are painted by gray and could not be passed.

\item For the last row (from bottom up), all the cells are painted
white, which means that all cells can be passed. The motivation is
to make sure that all the variable cells can connect to the final
cell (the upper-right cell) in the grid graph, which corresponds to
the destination node $t$ in ISCPCS.


\item There are a total of $(3m+2n)$ columns in the grid graph. If there are
$k$ clause vertices connecting to a fixed variable vertex in ISCPCS,
then there are $k$ clause cells connecting to a variable cell (say,
$C[i,j]$). The $k$ clause cells are located from the $(j-k)$-th
column to the $(j-1)$-th column, and the order of these columns are
adjusted to make those $k$ clause cells arranged from lower-left to
upper-right. (For an example, see Figure~\ref{fig:spread_fff}.) This
unique design can ensure that any monotone path from $C[1,1]$ to
$C[n+m+3,3m+2n]$ has to pass one clause cell and one variable cell
connect to it.
\end{enumerate}

\begin{figure}
 \centering
 \includegraphics{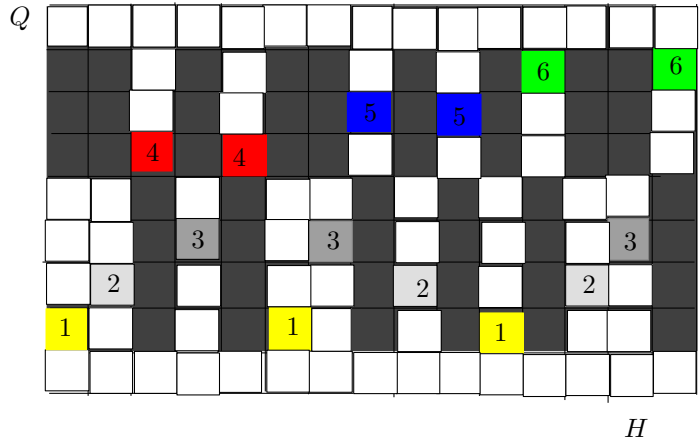}
 \caption{Illustration of the equivalence relation between the free space diagram (grid graph) and ISCPCS in Figure~\ref{fig:Np_proof}. The
horizontal coordinates denote the precise points, while the vertical
coordinates denote the imprecise points. A white cell means it can
be passed, while the gray cell means they could not be passed, and
the cells painted by the same color (and with the same number) means
any one of them can be passed.}
  \label{fig:spread_fff}
\end{figure}

\subsection{Realizing the grid graph geometrically}

To complete the proof that deciding $F^{\max}(Q,H)\leq\epsilon$ is
NP-hard, we need to construct a precise vertex sequence
$H=(h_1,h_2,...,h_{3m+2n})$ and an imprecise vertex sequence
$Q=(q_1,q_2,q_3,...,q_{n+m+3})$ (where each imprecise vertex is
modeled as a rectangle) such that the free space grid graph
constructed above can be geometrically realized.

Throughout the remaining parts, let $C(a,r)$ (resp. $D(a,r)$) be a
Euclidean circle (resp. disk) centered at $a$ and with radius $r$.
The rectangles used to model imprecise points do not need to be
along the same direction. The general idea of realizing the grid
graph geometrically is as follows.
\begin{figure}
 \centering
 \includegraphics[width=0.6\textwidth]{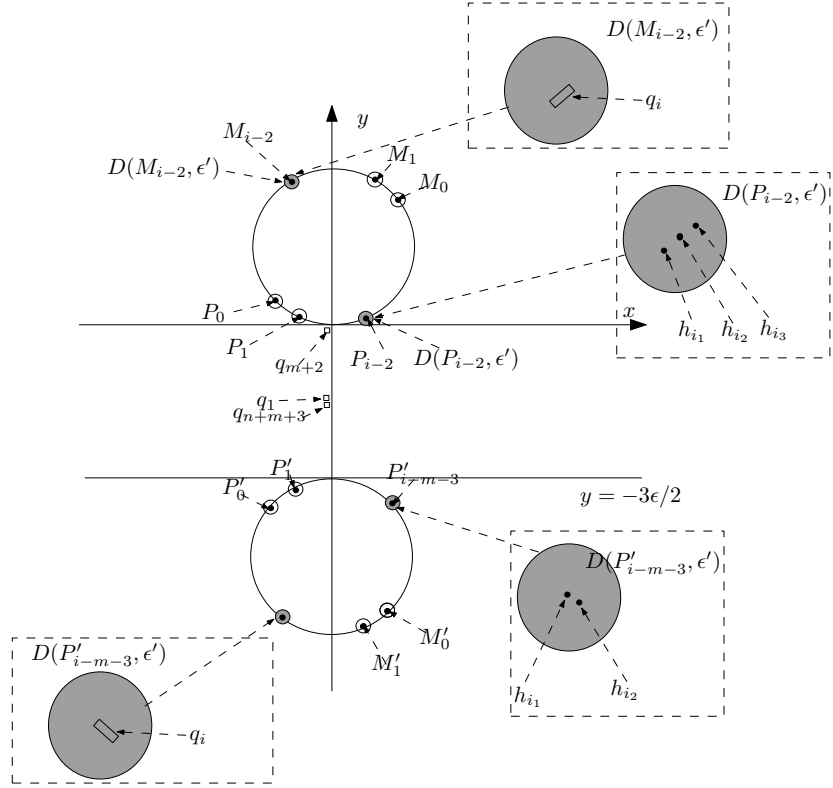}
 \caption{Illustration of the general idea of realizing the grid
graph geometrically.}
  \label{fig:show}
\end{figure}

\begin{enumerate}
\item For the 2nd to the $(m+1)$-th rows of the grid graph, we design the
points $P_i$, $M_i,i=0,1,2,...,m-1$, which satisfy
$d(P_i,M_i)=\epsilon$ and $d(P_i,M_j)\leq
\epsilon-2\epsilon'<\epsilon$ ($\epsilon'\ll\epsilon$) when $i \neq
j$. Each clause gadget is composed of three precise vertices
$h_{i_1},h_{i_2},h_{i_3}$ ($i_1,i_2,i_3$ is the index of sequence
$H$) and an imprecise vertex $q_i$ as in
Figure~\ref{fig:spread_variable}(b). $h_{i_1},h_{i_2},h_{i_3}$ are
located inside $D(P_{i-2},\epsilon')$, $q_i$  is located inside
$D(M_{i-2},\epsilon')$ when ($2 \leq i \leq m+1$). All the points
$P_i, i=0,1,2,...,m-1$ are located in a small region with diameter
less than $\epsilon/10$.

\item For the $(m+3)$-th to the $(n+m+3)$-th rows, we design the
points $P'_i$, $M'_i, i=0,1,2,...,n-1$, which satisfy
$d(P'_i,M'_i)=\epsilon$ and $d(P'_i,M'_j)\leq
\epsilon-2\epsilon'<\epsilon$ when $i \neq j$. Each variable gadget
is composed of two  precise vertices $h_{i_1},h_{i_2}$ and an
imprecise vertex $q_i$ as in Figure~\ref{fig:spread_variable}(a).
$h_{i_1},h_{i_2}$ are located inside $D(P'_{i-m-3},\epsilon')$,
$q_i$
 is located inside
$D(M'_{i-m-3},\epsilon')$ when $m+3 \leq i \leq n+m+2$. Again, all
the points $P'_i, i=0,1,2,...,n-1$ are located in another small
region with diameter less than $\epsilon/10$.

\item $d(P_i, M'_j)>3\epsilon/2>\epsilon, d(P'_i, M_j)>3\epsilon/2>\epsilon$. $d(p,q)\leq
\epsilon$ when $p\in D(P_i,\epsilon'), q\in D(P_j,\epsilon'), i \neq
j$, and $d(p,q)\leq \epsilon$ when $p\in D(P'_i,\epsilon')$, $q\in
D(P'_j,\epsilon')$, and $i \neq j$.

\item For the first and last row, the first imprecise vertex $q_1$ and last
imprecise vertex $q_{n+m+3}$ are located in a region
$D((0,-3\epsilon/4),\epsilon')$ which is fully covered by any circle
$C(p,\epsilon)$ where $p\in \bigcup D(P_i,\epsilon')$
($i=0,1,...,m-1$) and circle $C(p,\epsilon)$ where $p\in \bigcup
D(P'_i,\epsilon')$($i=0,1,...,n-1$).

\item For the $(m+2)$-th row, the vertex $q_{m+2}$ is located inide a
region $D((0,0),\epsilon')$, which is fully covered by the circle
$C(p,\epsilon)$ where $p\in \bigcup D(P_i,\epsilon')$
($i=0,1,...,m-1$) but not covered by any circle
  $C(p,\epsilon)$ where $p\in
\bigcup D(P'_i,\epsilon')$($i=0,1,...,n-1$).
\end{enumerate}

Due to space constraints, the details for realizing the grid graph
are given in the appendix (Section 7.2).

\begin{theorem}
Computing the upper bound of the discrete Fr\'{e}chet-distance with
imprecise input is NP-hard.
\end{theorem}

\begin{proof}
From our construction and Lemma 1, $F^{max}(Q,H)> \epsilon$ if and
only if there exist a choice that choose exactly one passable cell
of each color such that there is no monotone path from lower-left
cell to the upper-right cell in the equivalent free space grid
graph, which holds on if and only there exist an induced subgraph
$G_s$ consist of exactly one vertex of each color in equivalent
colored graph $G$ such that in $G_s$
 there is no monotone path from
$s$ to $t$, which in turn is true if and only if the corresponding
3SAT instance is satisfiable. The total reduction time is
$O((m+n)^2)$, and the theorem is proven. \qed

\end{proof}

\section{The discrete  Fr\'{e}chet distance with shortcuts for imprecise input}

As covered in the introduction, the discrete \frechet\ distance is
sensitive to local errors; hence, in practice, it makes sense to use
the {\em discrete \frechet\ distance with shortcuts}
\cite{Avraham2014}. This is defined as follows. (We comment that
this idea of taking shortcuts was used as early as in 2008 for
simplifying protein backbones \cite{BeregJWYZ08}.)

\begin{definition}
One-sided discrete Fr\'{e}chet distance with shortcuts: For two
point sequences $A=(a_1,a_2,a_3,...,a_n)$, and
$B=(b_1,b_2,b_3,...,b_m)$, let $F_c(A,B)$ denote the discrete
Fr\'{e}chet distance with shortcuts on side $B$, where\\
$F_c(A,B)=\min \{ F(A,B')\}$ and $B'$ is a non-empty subsequence of
$B$.
\end{definition}

Alternatively, we can define the discrete Fr\'{e}chet distance with
shortcuts on side $B$ as follows. We loosely call each edge
appearing in the set $E_\delta$ (in Def.~1) a {\em match}. Given a
match $(a_i,b_j)$, the next match $(a_k,b_l)$ needs to satisfy one
of the three conditions:

a) $k=i+1,l=j$;

b) $k=i,l>j$;

c) $k=i+1,l>j$.

In \cite{Avraham2014}, Avraham \emph{et al.} gave a definition of
discrete Fr\'{e}chet distance with shortcuts, they assumed no
simultaneous jumps on both sides (i.e., case c) does not occur),
though they claimed that their algorithm can be easily extended to
this case when simultaneous jumps are allowed.

Now we define the discrete Fr\'{e}chet distance  with shortcuts for
imprecise data as follows:

\begin{definition} $F^{\max}_1(U,W)$:
For two region sequences $U=(u_1,u_2,...,u_n)$ and
$W=(w_1,w_2,...,w_m)$, the upper bound of the discrete Fr\'{e}chet
distance with shortcuts on side $W$ is defined as
$F^{\max}_1(U,W)=\max \{F_c(A,B)\}$, where $A=(a_1,a_2,...,a_n)$
(resp. $B=(b_1,b_2,...,b_m$) is a possible realization of $U$ (resp.
$W$) satisfying $a_i\in u_i$ and $b_j\in w_j$.
\end{definition}

\subsection{Computing $F^{\max}_1(U,W)$ when one sequence is imprecise}

At first, we consider the case when $U$ is a precise vertex sequence
composed of $n$ precise points in $R^d$, and $W$ is an imprecise
vertex sequence, where each of the $m$ imprecise points is modeled
as a ball in $R^d$.

Let $u'_i$ denote the ball centered at $u_i$ with radius $\delta$,
i.e., $u'_i=D(u_i,\delta)$. Let $M(i,j)$ denote the {\em match} or
{\em matching pair} between $u_i$ and $w_j$.

For the discrete Fr\'{e}chet distance with shortcuts on side $W$, we
only need to consider the jump from $M(i,j)$ to $M(i+1,k)$ ($k\geq
j$), and there is no need to consider the jump from $M(i,j)$ to
$M(i,l)$($l\geq j$) . This is due to that the match $M(i,l)$ will
jump to $M(i+1,l')$ ($l'\geq l$) finally when $i < n$, and we can
jump directly from $M(i,j)$ to $M(i+1,l')$ without passing through
$M(i,l)$.

The algorithm to decide $F^{\max}_1(U,W)\leq \delta$ is as follows.

Starting from the starting matching pair $M(1,j^*(1))$ to the ending
matching pair $M(n,j^*(n))$ if possible, where $j^*(1)$ is the
smallest $k$ ($1 \leq k\leq m$) which satisfies that $w_{k}
\subseteq u'_{1}$, $j^{*}(i)$ be the  index of sequence $W$ computed
by the decision procedure  $F^{\max}_1(U,W) \leq \delta$ below for
each fixed $i$, let $S$ denote the set of those matches (or matching
pairs) $M(i,j^{*}(i))$.


\vskip 2pt
\setlength{\fboxsep}{10pt}%

\fbox{%

\begin{Bflushleft}[b]

1. $i=1,j=1$.\\

2. While ($i \leq n$)\\

 \ \ \ \ \ \ \ \ \ \ \ \ Find a  smallest $k$ ($j \leq k\leq m$) which satisfies that $w
_{k} \subseteq u'_{i}$.\\

 \ \ \ \ \ \ \ \ \ \ \ \ If  $k$  exists,
 let $j^{*}(i)=k$, add the match $M(i,j^{*}(i))$  to $S$,\\

 \ \ \ \ \ \ \ \ \ \ \ \ \ \ and update $j=k$, $i=i+1$.\\

\ \ \ \ \ \ \ \ \ \ \ \ Else return $F^{\max}_1(U,W)>\delta$.\\

Return $F^{\max}_1(U,W)\leq \delta$.

\end{Bflushleft}

} \vskip 2pt

\begin{center}
 {\bf Fig.7}  The decision procedure for $F^{\max}_1(U,W) \leq \delta$, where
$U$ is a precise sequence and $W$ is an imprecise sequence.
\end{center}

\vskip 4pt

We will show that the above procedure correctly decides whether
$F^{\max}_1(U,W) \leq \delta$.

\begin{lemma}

\label{lemma:monotone} There exists a  realization of $W$ to make
$j^{*}(i)$ be the smallest index of sequence $W$ such that
$M(i,j^{*}(i))$ is reachable by jump from $M(1,j^*(1))$ for each
fixed $i$. That means there is no monotone increasing path from
$M(1,j^*(1))$ to $M(i,j)$ when $1 \leq j<j^{*}(i)$.
\end{lemma}

\begin{proof}

We prove this lemma by an induction on $i$.

(1) Basis: When $i=1$, then there exists a realization
$(b_1,b_2,...,b_{j^*(1)-1})$ of  $(w_1,w_2,...,w_{j^*(1)-1})$
respectively which satisfies $d(u_1,b_j)>\delta$ for $1 \leq
j<j^*(1)$. Hence, there exists a realization of $W$ which make the
matching $M(1,j)$ is not reachable when $1 \leq j<j^*(1)$.

(2) Inductive hypothesis: We assume that there exists a realization
$(b_1,b_2,...,b_{j^*(k)-1})$ of $(w_1,w_2,...,w_{j^*(k)-1})$ which
makes $M(k,j)$ ($1 \leq j<j^*(k)$)  not reachable when $i=k$.

(3) Inductive step: We consider the case when $i=k+1$. For
$M(k+1,j)$, $j^{*}(k) \leq j < j^{*}(k+1)$ not in $S$ then there
exists a realization $(b_{j^{*}(k)},...,b_{j^{*}(k+1)-1})$ of
$(w_{f(k)},...,w_{j^{*}(k+1)-1})$ which satisfies
$d(u_{k+1},b_j)>\delta$, for $j^{*}(k)\leq j<j^{*}(k+1)$. Based on
the inductive hypothesis, there exists a realization
$(b_1,b_2,...,b_{j^{*}(k)})$ which makes $M(k,j)$, $1 \leq
j<j^{*}(k)$, not reachable. By combining the two parts, there exists
a realization $(b_1,b_2,...,b_{j^{*}(k)},...,b_{j^{*}(k+1)-1})$ of
$(w_1,w_2,...,w_{j^{*}(k)},...,w_{j^{*}(k+1)-1})$ which makes
$M(k+1,j)$, where $1 \leq j<j^{*}(k+1)$, not reachable. \qed
\end{proof}

\begin{lemma}
In $\reals^d$, given a precise vertex sequence $U$ with size $n$ and
an imprecise vertex sequence $W$ with size $m$ (each modeled as a
$d$-ball), whether $F^{\max}_1(U,W) \leq \delta$ can be determined
in $O(d(n+m))$ time and space.
\end{lemma}

\begin{proof}
If the above decision procedure returns ``$F^{\max}_1(U,W)>\delta$"
when $i=k+1$, then there exists a realization of $W$ to make
$M(k+1,j)$ not reachable, for $k+1<n$ and $j^{*}(k) \leq j\leq m$,
based on Lemma~\ref{lemma:monotone}. That means
$F^{\max}_1(U,W)>\delta$.

If the above decision procedure returns ``$F^{\max}_1(U,W)\leq
\delta$", then there exists a monotone path from $M(1,j^{*}(1))$ to
$M(n,j^{*}(n))$ for any realization of $W$. The reason is that, for
any $i$ and $j^{*}(i)$, $M(i,j^{*}(i))\in S$ implies $w_{j^{*}(i)}
\subseteq u'_{i}$, which means $F^{\max}_1(U,W)\leq \delta$. The
correctness is hence proven.

As for the running time, checking whether $w_{j} \subseteq u'_{i}$
takes $O(d)$ time. The decision procedure incrementally tests on a
row- and column-monotone path. Therefore, it runs in $O(d(m + n))$
time and space. \qed
\end{proof}

Let $\delta_{i,j}=d(u_i,c_j)+r_j$ where $c_j$ and $r_j$ are the
center and radius of $w_j$ respectively. For the optimization
problem, there are a total of $O(mn)$ events when $\delta$ increases
continuously. Here, an event $w_{j} \subseteq u'_{i}$ occurs when
$\delta$ increases to $\delta_{ij}$. Therefore, we can solve the
optimization problem of computing $F^{\max}_1(U,W)$ in $O(dmn\log
mn)$ time by sorting $\delta_{i,j}$'s and performing a binary
search. We show how to improve this bound below. We first consider
the planar case in the next theorem.

\begin{theorem}
In $\reals^2$, given a precise vertex sequence $U$ with size $n$ and
an imprecise vertex sequence $W$ with size $m$, all modeled as disks
in $\reals^2$ with an equal radius, $F^{\max}_1(U,W)$ can be
computed in $O((m^{2/3}n^{2/3}+m+n)\log^3(m+n))$ time.
\end{theorem}

\begin{proof}
The event $w_{j} \subseteq u'_{i}$ occurs when
$\delta=\delta_{i,j}$. We do not need to sort the $O(mn)$ distances;
instead, we can use the distance selection algorithm in
\cite{Katz1997} as follows. One can select the $k$-th smallest
pairwise distance $d_k$ in $A \times B$, where $A$ and $B$ are two
precise vertex sequences in the plane and $|A|=n,|B|=m$. The running
time of this distance selection algorithm is
$O((m^{2/3}n^{2/3}+m+n)\log^2(m+n))$ \cite{Katz1997}. By combining
this distance selection algorithm and the binary search, we can
compute $F^{\max}_1(U,W)$ in $O((m+n)\log
(mn)+(m^{2/3}n^{2/3}+m+n)\log^2(m+n)\log (mn))$
$=O((m^{2/3}n^{2/3}+m+n)\log^3(m+n))$ time. \qed
\end{proof}

Unfortunately, the distance selection algorithm could not be
extended to high dimensional space. Hence, in a dimension higher
than two, we use a dynamic programming method to compute
$F^{\max}_1(U,W)$.

Let $U(i,j)$ (resp. $W(i,j)$) denote the partial sequence
$U(i,j)=(u_i,u_{i+1},...,u_j)$ (resp.
$W(i,j)=(w_i,w_{i+1},...,w_j))$. $F^{\max}_1(i,j)$ denotes the upper
bound of the discrete Fr\'{e}chet distance with shortcuts on side
$W(1,j)$ for sequences $U(1,i)$ and $W(1,j)$, and $Z^{\max}_1(i,j)$
denotes the upper bound of the discrete Fr\'{e}chet distance with
shortcuts on side $W(1,j)$ between $U(1,i)$ and $W(1,j)$ on the
condition that $w_j$ is retained (not cut), that means $w_j$ match
$u_i$, namely
$Z^{\max}_1(i,j)=\max\{F^{\max}_1(i-1,j),\delta_{i,j}\}$. While in
$F^{\max}_1(i,j)$, $u_i$ may do not   match $w_j$ as $w_j$ may be
cut.

Then we have the recurrence relations as follows.

\[Z^{\max}_1(i+1,j) \\
=\left\{\begin{array}{ll}
\max\{F^{\max}_1(i,j),\delta_{i+1,j}\},& \mbox{$i> 0$}\\
\delta_{1,j}, & \mbox{$i=0$}
\end{array}
\right. \]

\[F^{\max}_1(i,j+1)\\
=\left\{\begin{array}{ll}
\min\{F^{\max}_1(i,j),Z^{\max}_1(i,j+1)\}, & \mbox{$j> 0$}\\
Z^{\max}_1(i,j+1), & \mbox{$j=0$}
\end{array}
\right. \]

We need to run these two recurrence relations alternatively, e.g.
computing $Z^{\max}_1(i,*)$ ($(i,*)$ denotes $\{(i,j),1 \leq j \leq
m\}$), then $F^{\max}_1(i,*)$, then $Z^{\max}_1(i+1,*)$, etc. It is
easy to see that this dynamic programming algorithm takes $O(mn)$
time and space after all the distances $\delta_{i,j}$ are calculated
in $O(dmn)$ time and space. On the other hand, the space complexity
can be improved as we only need to store a constant number of
columns of values and compute $\delta_{i,j}$ when needed. Hence we
have the following theorem.

\begin{theorem}
\label{theorem:dynamic} In $\reals^d$, given a precise vertex
sequences $U$ of size $n$ and an imprecise vertex sequences $W$ of
size $m$ (each modeled as a $d$-ball), $F^{max}_1(U,W)$ can be
computed in $O(dmn)$ time and $O(d(m+n))$ space.
\end{theorem}

\subsection{Computing $F^{\max}_1(U,W)$ when both sequences are imprecise}

In this subsection, we consider the problem of computing
$F^{\max}_1(U,W)$ when both $U=(u_1,u_2,...,u_n)$ and
$W=(w_1,w_2,...,w_m)$ are imprecise sequences, where each vertex is
modeled as a disk in $\reals^2$. (Our algorithm works in $\reals^d$,
but as it involves Voronoi diagram in $\reals^d$, the high cost
makes it impractical.)

For two imprecise sequences $U,W$, and a precise point $p$, the
maximal distance between a  precise point $p$ and a region $w_x$ is
defined as
 $D_{\max}(p,w_x)=\max\{d(p,q),q\in w_x\}$. Let
$D_{\min}(p,j,k)=\min_{j\leq x\leq k}\{D_{\max}(p,w_x)\}$ denote the minimal
distance between a point $p$ and several regions
$\{w_j,w_{j+1},...,w_k\}$.

We define $D(i,j,k)=\max\{ D_{\min}(p,j,k), p \in u_i\}$.  In this
subsection, we compute the $j^*(i)$ by using the decision procedure
below:

\vskip 2pt

\setlength{\fboxsep}{10pt}%

\fbox{%
\begin{Bflushleft}[b]

1. $i=1,j=1$.\\

2. While ($i \leq n$)\\

 \ \ \ \ \ \ \ \ \ \ \ \ Find a  smallest $k$ ($j \leq k\leq m$) which satisfies  $D(i,j,k) \leq
 \delta$.\\

 \ \ \ \ \ \ \ \ \ \ \ \ If  $k$  exists, let $j^*(i)=k$, add $M(i,j^*(i))$ to $S$,\\

 \ \ \ \ \ \ \ \ \ \ \ \ \ \  and update $j=k$, $i=i+1$.\\

\ \ \ \ \ \ \ \ \ \ \ \ Else return $F^{\max}_1(U,W)>\delta$.\\

Return $F^{\max}_1(U,W)\leq \delta$.\\

\end{Bflushleft}}
\vskip 2pt

\begin{center}
{\bf Fig.~8} The decision procedure for $F^{\max}_1(U,W) \leq
\delta$ when both $U$ and $W$ are imprecise sequences.
\end{center}

\vskip 4pt

\begin{lemma}
Given two imprecise vertex sequences $U$ and $W$ with sizes $|U|=n$
and $|W|=m$, each vertex modeled as a disk in $\reals^2$),
$F^{\max}_1(U,W) \leq \delta$ can be determined in $O(m^2+n)$ time
and $O(m+n)$ space.
\end{lemma}

\begin{proof}
The correctness is given as follows.

(1) If the decision procedure returns ``$F^{max}_1(U,W)>\delta$",
then there exists a realization of $U$ and $W$ which makes it
impossible to reach the last matching pair $M(n,j^*(n))$. The
argument is similar to Lemma~\ref{lemma:monotone} and omitted here.
That means $F^{max}_1(U,W)>\delta$.

(2)  If the decision procedure returns ``$F^{max}_1(U,W)\leq
\delta$", then $S$ has $n$ elements, i.e., $S=\{M(1,j^*(1))$,
$M(2,j^*(2))$, $M(3,j^*(2))$,...,$M(n,j^*(n))\}$, where $j^*(i) \leq
j^*(i+1)$. We claim that there exists  monotone matching pair set
$\{M(i,j(i))| 1 \leq i\leq n\}$  ($j(i) \leq j^*(i)$) under any
realization of $U$ and $W$, where $j(i)$ is an index of W and $u_i$
match $w_{j(i)}$ and $j(i)\leq j(i+1) $.

We prove the claim by an induction on $i$.

(2.1) Basis: When $i=1$, if all the matching pairs $M(1,j)$, $j \leq
j^*(1)$, are not reachable, then there exists a realization $b_x, 1
\leq x \leq j^{*}(1)$, and $a_1$ which satisfy $d(a_1,b_x)>\delta$.
Then $D(1,1,j^{*}(1)) > \delta$, and we have a contradiction, that
means there exist $j(1)\leq j^{*}(1)$ such that $M(1,j(1))$ is
possible under any realization.

(2.2) Inductive hypothesis:  We assume that the claim holds when
$i=l$.

(2.3) Inductive step: Now we consider the case when $i=l+1$. By the
inductive hypothesis, there exists monotone matching set
$M(1,j(1)),M(2,j(2)),...,M(l,j(l))$, $j(i) \leq j^*(i)$ under any
realization of $U(1,l), W(1,j)$. As $D(l+1,j,j^*(l+1)) \leq
D(l+1,j^*(l),j^*(l+1)) \leq \delta$, there exists a matching pair
$M(l+1,j(l+1))$, where $j(l+1) \leq j^*(l+1)$, which is reachable by
jumping directly from $M(l,j(l))$ under any realization of
$U(l+1,l+1)$ and $W(j+1,j(l+1))$. Hence, if the decision procedure
returns ``$F^{max}_1(H,Q)\leq \delta$", then $F^{max}_1(H,Q)\leq
\delta$.

\vskip 6pt

We now compute the time it takes to find a smallest $k$ ($j \leq
k\leq m$) satisfying $D(i,j,k) \leq \delta$. The steps to compute
$D(i,j,k)$ can be done as follows.

(I) We compute the inverted additive Voronoi
Diagram~\cite{loffler2009} (iaVD) of imprecise vertices
$w_j,w_{j+1},...,w_{k}$ modeled as disks (may have different sizes),
which takes $O((k-j)\log (k-j))$ time.

(II) If the imprecise region $u_i$ intersects the boundary of iaVD,
then some vertex of the partial boundary within $u_i$ would be the
realization of $u_i$ in computing $D(i,j,k)$. Otherwise, the
imprecise region $u_i$ is located in the cell controlled by some
site $w_x$. Then, the diameter of the region $u_i\cup w_x$ would be
$D(i,j,k)$. This step takes $O(k-i)$ time.

As we need to construct the inverted additive Voronoi Diagram
incrementally, each single insertion takes $O(s)$ time, where $s$ is
the size of the iaVD. Hence the total time to find a smallest $k$
($j \leq k\leq m$) is $\sum_{j\leq x\leq k} (x-j)=O((k-j)^2)$.
Therefore, the total time complexity is $\sum_{1 \leq i \leq n}
(j^*(i+1)-j^*(i))^2=O(m^2+n)$. \qed
\end{proof}

For the optimization problem, we again use a dynamic programming
algorithm to solve it. The algorithm is similar to that in
Theorem~\ref{theorem:dynamic} and the difference is to use
$D(i,j,k)$ instead of $\delta_{i,j}$. The recurrence relation is as
follows.

\[F^{\max}_1(i+1,k)\\
=\left\{\begin{array}{ll}
\min_{1\leq j\leq k}\{\max\{F^{\max}_1(i,j),D(i+1,j,k)\}\}, & \mbox{$i> 0$}\\
D(1,1,k), & \mbox{$i=0$}
\end{array}
\right. \]

It seems that the dynamic programming algorithm takes $O(nm^2)$ time
after all the distances $D(i,j,k)$ are calculated in $O(nm^3)$ time.
However, we can use the Monge property to speed up the computation
of dynamic programming, we only need to compute $O(nm)$ distances
$D(i,j,k)$'s in $O(nm^2)$ time.

(1) $F^{\max}_1(i,j)$ is a monotone decreasing function when $j$
increases for a fixed $i$.

(2) $D(i+1,j,k)$ is a monotone increasing function when $j$
increases for fixed $i$ and $k$.

(3) Let $j_k$ denote the index satisfying $F^{\max}_1(i+1,k)=\max
\{F^{\max}_1(i,j_k),D(i+1,j_k,k)\}$, and $j_{k+1}$ denote the index
satisfying $F^{\max}_1(i+1,k+1)=$ $\max \{F^{\max}_1(i,j_{k+1})$,
$D(i+1,j_{k+1},k+1)\}$,  then $j_{k+1} \geq j_{k}$.

Hence we only need to try distances $D(i+1,j_{k},k+1),
D(i+1,j_{k}+1,k+1), D(i+1,j_{k}+2,k+1),...,D(i+1,j_{k+1},k+1),
D(i+1,j_{k+1}+1,k+1)$ when computing $F^{\max}_1(i+1,k+1)$ for a
fixed $i$ and $k$, namely $(j_{k+1}-j_{k}+2)$ distances, hence the
total number of distance is $O(m)$ for a fixed $i$.

Hence $F^{\max}_1(i+1,k)$ ($1 \leq k \leq m$) can be calculated in
$O(m)$ time after the distances $D(i+1,j,k)$ ($1 \leq k \leq m$) are
calculated for a fixed $i$. We then only need to try $O(m)$
distances $D(i+1,j,k)$ ($1 \leq k \leq m$): the update of iaVD needs
at most $O(m)$ insert operations, $O(m)$ deletion operations, and
$O(m)$ query operations, each takes at most $O(m)$ time. Hence the
total time is $O(m^2)$ for a fixed $i$. Hence we have the theorem
below.

\begin{theorem}
In $\reals^2$, given two imprecise sequences $U$ and $W$ of size $n$
and $m$ respectively, where each imprecise vertex is modeled as a
disk, $F^{max}_1(U,W)$ can be computed in $O(nm^2)$ time.
\end{theorem}

We comment that when both $U$ and $W$ are imprecise, our algorithm
could still work in $\reals^d$. But due to the high cost (like
constructing the $d$-dimensional Voronoi diagram), the algorithm
then becomes impractical. Hence, we only focus on the problem in
$\reals^2$ for this case.

\section{Concluding remarks}

In this paper, we consider the problem of computing the discrete
Fr\'{e}chet distance of imprecise input. We  address the open
problem posed by Ahn~\cite{AhnIsaac2010,Ahn2012}  \emph{et al.} a
few years ago, and show that the discrete Fr\'{e}chet distance upper
bound problem of imprecise data is NP-hard. And our NP-hardness
proof is quite complicate, the construction has a combinatorial and
a geometric part. In the combinatorial part, we interpret the
imprecise discrete distance in terms of finding monotone paths
through a colored grid graph; In the geometric part, we show that
the relevant colored free space diagram grids can be realized
geometrically.
 Given two imprecise
vertex sequence $U,W$ (each vertex modeled as a $d$-dimensional
ball), we show that the upper bound of the discrete Fr\'{e}chet
distance between $U$ and $W$ can be computed in polynomial time if
allowing shortcuts on one side. It would be interesting to consider
these problems under the continuous Fr\'{e}chet distance.

\section{Acknowledgments}
CF's research is supported by the Benjamin PhD Fellowship. BZ's
research is supported by the Open Fund of Top Key Discipline of
Computer Software and Theory in Zhejiang Provincial Colleges at
Zhejiang Normal University.

\bibliographystyle{abbrv}

\begin{thebibliography}{1}

\bibitem{Abellanas2001}
M. Abellanas, F. Hurtado, C. Icking, R. Klein, E. Langetepe, L. Ma,
B. Palop, and V. Sacristan.
\newblock {Smallest color-spanning objects}.
\newblock {\em Proc. 9th European Sympos. Algorithms (ESA'01)}, pp. 278--289, 2001.

\bibitem{AgarwalAKS14}
P. Agarwal, R. Avraham, H. Kaplan, and M. Sharir.
\newblock Computing the discrete {F}r{\'{e}}chet distance in subquadratic time.
\newblock {\em {SIAM} J. Comput.}, 43(2):429-449, 2014.

\bibitem{AhnIsaac2010}
H-K.~Ahn, C.~Knauer, M.~Scherfenberg, L.~Schlipf, and A.~Vigneron.
\newblock Computing the discrete Fr\'{e}chet distance with imprecise input.
\newblock In {\em Proc. ISAAC'10}, pp.~422-433, 2010.

\bibitem{Ahn2012}
H-K.~Ahn, C.~Knauer, M.~Scherfenberg, L.~Schlipf, and A.~Vigneron.
\newblock Computing the discrete Fr\'{e}chet distance with imprecise input.
\newblock {Int. J. Comput. Geometry Appl.}, 22(1):27-44, 2012.

\bibitem{AltB05}
H.~Alt and M.~Buchin.
\newblock Semi-computability of the Fr{\'e}chet distance between surfaces.
\newblock In {\em Proc. EuroCG'05}, pp.~45-48, 2005.

\bibitem{AltERW03}
H.~Alt, A.~Efrat, G.~Rote, and C.~Wenk.
\newblock Matching planar maps.
\newblock In {\em Proc. SODA'03}, pp.~589-598, 2003.

\bibitem{AltG95}
H.~Alt and M.~Godau.
\newblock Computing the Fr{\'e}chet distance between two polygonal curves.
\newblock {\em Int. J. Comput. Geometry Appl.}, 5:75-91, 1995.

\bibitem{Avraham2014}
R.~Avraham, O.~Filtser, H.~Kaplan, M.~Katz, and M.~Sharir.
\newblock The discrete Fr\'{e}chet distance with shortcuts via approximate distance counting and selection.
\newblock In {Proc. SoCG'14}, pages 377, 2014.

\bibitem{BeregJWYZ08}
S.~Bereg, M.~Jiang, W.~Wang, B.~Yang, and B.~Zhu.
\newblock Simplifying 3{D} polygonal chains under the discrete {F}r{\'{e}}chet distance.
\newblock In {\em Proc. 8th Latin American Theoretical Informatics Sympos. (LATIN'08)}, pp.~630--641, 2008.

\bibitem{BPSW05}
S.~Brakatsoulas, D.~Pfoser, R.~Salas, and C.~Wenk.
\newblock On map-matching vehicle tracking data.
\newblock In {\em Proc. the 31st International Conf. on
  Very Large Data Bases (VLDB'05)}, pp.~853-864, 2005.

\bibitem{BBW08}
K.~Buchin, M.~Buchin, and C.~Wenk.
\newblock Computing the fr\'{e}chet distance between simple polygons.
\newblock {\em Comput. Geom. Theory Appl.}, 41(1-2):2-20, 2008.

\bibitem{Das2009}
S.~Das, P.~Goswami, and S.~Nandy.
\newblock  Smallest Color-Spanning Object Revisited.
\newblock {\em Int. J. Comput. Geometry Appl}, 19(5):457-478, 2009.

\bibitem{Driemel2013}
A. Driemel and S. Har-Peled.
\newblock Jaywalking your dog: Computing the Fr\'{e}chet distance with shortcuts,
\newblock {\em SIAM J. Comput.}, 42(5):1830-1866, 2013.

\bibitem{Wien94computingdiscrete}
T.~Eiter and H.~Mannila.
\newblock Computing discrete fr{\'e}chet distance.
\newblock Technical report, Technische Universitat Wien, 1994.

\bibitem{Chenglin15}
C.~Fan, O.~Filtser, M.~Katz, T.~Wylie, and B.~Zhu.
\newblock On the chain pair simplification problem.
\newblock In {Proc. WADS'15}, LNCS 9214, pp.~351-362, 2015.

\bibitem{chenglin2011}
C.~Fan, J.~Luo, and B.~Zhu.
\newblock Fr\'{e}chet-distance on road networks.
\newblock In {\em Proc. CGGA'10}, LNCS 7033, pp.~61-72, 2011.



\bibitem{JiangXZ08}
M. Jiang, Y. Xu, and B. Zhu.
\newblock Protein structure-structure alignment with discrete Fr{\'{e}}chet
  distance.
\newblock {\em J. Bioinfo. and Comput. Biology}, 6(1):51-64, 2008.


\bibitem{wenqi2012}
W.~Ju, J.~Luo, B.~Zhu, and O. Daescu.
\newblock Largest area convex hull of imprecise data based on axis-aligned squares.
\newblock {\em J. Comb. Optim.}, 26(4):832--859, 2013.

\bibitem{Katz1997}
M. J. Katz and M. Sharir.
\newblock An expander-based approach to geometric optimization,
\newblock SIAM J. Comput., 26(5):1384-1408, 1997.

\bibitem{Khanban2003}
A.~A.Khanban and A.~Edalat.
\newblock Computing Delaunay triangulation with imprecise input data.
\newblock In {\em Proc. 15th Canadian Conference on Computational Geometry (CCCG'03)}, pp. 94-97, 2003.

\bibitem{loffler2009}
C.~Knauer, M.~L{\"o}ffler, M.~Scherfenberg, and T.Wolle.
\newblock The directed Hausdorff distance between imprecise point sets.
\newblock In {\em Proc. ISAAC'09}, LNCS 5878, pp. 720-729, 2009.

\bibitem{loffler2010}
M.~L{\"o}ffler and J.~Snoeyink.
\newblock Delaunay triangulation of imprecise points in linear time after preprocessing.
\newblock {\em Comput. Geom. Theory Appl.}, 43(3):234-242, 2010.

\bibitem{loffler2006}
M.~L{\"o}ffler and M.J.~van Kreveld.
\newblock Largest and smallest tours and convex hulls for imprecise points.
\newblock  In {\em Proc. SWAT'06}, LNCS 4059, pp. 375-387, 2006.

\bibitem{Maheshwari2011}
A.~Maheshwari, J.~Sack, K.~Shahbaz, and H.~Zarrabi-Zadeh.
\newblock Fr\'{e}chet distance with speed limits.
\newblock {\em Comput. Geom. Theory Appl.}, 44(2):110-120, 2011.

\bibitem{MaheshwariSSZ11}
A.~Maheshwari, J.~Sack, K.~Shahbaz, and H.~Zarrabi-Zadeh.
\newblock Staying close to a curve.
\newblock {\em Proc. CCCG'11}, pp.~55-58, 2011.

\bibitem{Rote07}
G.~Rote.
\newblock Computing the Fr\'{e}chet distance between piecewise smooth curves.
\newblock {\em Comput. Geom. Theory Appl.}, 37(3):162-174, 2007.

\bibitem{Evans2008}
J.~Sember and W.~Evans.
\newblock Guaranteed Voronoi diagrams of uncertain sites.
\newblock  In {\em Proc. CCCG'08}, 2008.


\bibitem{Wylie2013}
T. Wylie and B. Zhu.
\newblock Protein chain pair simplification under the discrete Fr\'{e}chet distance,
\newblock {\em IEEE/ACM Trans. Comput. Biology Bioinform.}, 10(6):1372-1383, 2013.

\bibitem{Wylie2014}
T. Wylie and B. Zhu.
\newblock Following a curve with the discrete Fr\'{e}chet distance.
\newblock {\em Theoretical Computer Science}, 556:34-44, 2014.

\bibitem{Zhu07}
B.~Zhu.
\newblock Protein local structure alignment with discrete {F}r{\'{e}}chet
  distance.
\newblock {\em J. of Comput. Biology}, 14(10):1343-1351, 2007.


\end{thebibliography}

\newpage
\section{Appendix}

\setcounter{theorem}{1}

\subsection{Proof of Lemma 1}
\begin{proof}
Let $\phi$ be a Boolean formula in conjunctive normal form with $n$
variables $x_1$, $x_2$, $\ldots,x_n$ and $m$ clauses $C_1$, $C_2$,
$\ldots,C_m$, each of size at most three. We take the following
steps to construct an instance $G$ of ISCPCS.

For each Boolean variable $x_i,\overline{x_i}$ in $\phi$, we use two
vertices with the same color (and different variables always use
different colors), one denoted as $x_i$, the other denoted as
$\overline{x_i}$. See $x_1$, $x_2$, and $x_3$ in
Figure~\ref{fig:Np_proof} for example. Eventually, we have to pick
one of the two vertices to retain this color. One represents that
the variable $x_i$ is assigned the \textbf{True} value, and the
other corresponds to the value \textbf{False}.

For each clause $C_i$ in $\phi$, we construct three vertices with
the same color (which has never used before). (We use shapes instead
of colors in Figure~\ref{fig:Np_proof} to emphasize the difference
with variables; for example, we use three triangular vertices in
Figure~\ref{fig:Np_proof} to denote clause $(\overline{x_1}\vee
x_2\vee \overline{x_3})$). We then add directed edges between the
clause vertices and variable vertices, the rule is as follows: let
vertex $p_i$ be the vertex used to denote variable $\overline{x_i}$
(resp. $x_i$), and vertex $c_{i,j}$ be the vertex used to denote the
clause $C_j$ which contains $x_i$ (resp. $\overline{x_i}$), then we
add a directed edge from $c_{i,j}$ to $p_i$, see
Figure~\ref{fig:Np_proof}.

At last we add edges from the source  $s$ to each vertex denoting a
clause, and add edges from each vertex denoting a variable to the
destination $t$. It is easy to ensure that there is no crossing
between the edges: all the vertices denoting variables are arranged
from left to right, there is a clause vertices for each literal
(variable vertex), and each of the clause vertices connecting to a
fixed variable vertex $x_i$ is just below the variable vertices
$x_i$. Let the resulting directed geometric graph be $D$. We next
complete the proof by proving that $\phi$ is satisfiable iff there
exists an induced subgraph $G_s$ such that in $G_s$ there is no path
from $s$ to $t$.

`$\rightarrow$' If $\phi$ is satisfiable with some truth assignment,
then there is at least one true literal in each clause $C_i$. We
show how to compute the induced subgraph $G_s$ from $G$ as follows.
For each pair of variable vertices representing
$\{x_j,\overline{x_j}\}$, we pick one which is assigned {\bf True}.
Let $C_i=u \vee v \vee w$, where $u,v,w$ are literals in the form of
$x_j$ or $\overline{x_j}$. Then we have three clause vetices
$c_{u,i}, c_{v,i}$ and $c_{w,i}$, of the same color $color_i$,
representing the clause $C_i$. WLOG,  just suppose that $u$ is a
true literal in $C_i$ (pick any one true literal if there exist more
than one literal in $C_i$ be true), and we choose the vertex
$c_{u,i}$ to cover the color $color_i$. By construction, in $D$,
there is no edge from the vertex $c_{u,i}$ to the vertex $p_u$
representing $u$.
Hence, there is no  path from $s$ to $t$ crossing the clause vertex
$c_{u,i}$ (representing clause $C_i$). As this holds for all clauses
and any path  from $s$ to $t$ has to pass a clause vertex, hence
there exists an induced subgraph $G_s$ consist of exactly one vertex
of each color with no path from $s$ to $t$.

`$\leftarrow$' If there exists an induced subgraph $G_s$ consist of
exactly one vertex of each color with no path from $s$ to $t$, we
need to prove that $\phi$ is satisfiable. Suppose to the contrary
that $\phi$ is not satisfiable, then at least one clause is not
satisfiable. Let this clause be $C_i=u \vee v \vee w$, and let the
clause vertices $c_{u,i}$, $c_{v,i}$ and $c_{w,i}$ connect to the
variable vertices $p_u$, $p_v$ and $p_w$ in $D$ which correspond to
the variables $\overline{u}$, $\overline{v}$, $\overline{w}$
respectively. As $u$, $v$ and $w$ are all false, $p_u$, $p_v$, $p_w$
are all picked in the induced subgraph. Then, there exists a path
from $s$ to $t$ passing through $c_{u,i}$, $c_{v,i}$ or $c_{w,i}$,
as one of them must be picked. A contradiction!

Hence, $\phi$ is satisfiable if and only if there exists an induced
subgraph $G_s$ consist of exactly one vertex of each color  with no
path from $s$ to $t$. The reduction obviously takes $O(n+m)$ time.
\qed
\end{proof}

\subsection{Details for realizing the grid graph geometrically}
The details for realizing the grid graph are given here. First, we
create the points used for determining the position of the vertices
in $H$ and $Q$. Let $\theta$ be satisfying that
 $\max\{m,n\}*\theta \leq \pi/20$.
Let $N=\max\{m,n\}$, and WLOG, let $N$ be even. We construct a
circle $C(O,r)$, where $O=(0, \epsilon/2)$ and $r=\epsilon/2$.

\begin{figure}
 \centering
 \includegraphics[width=0.5\textwidth]{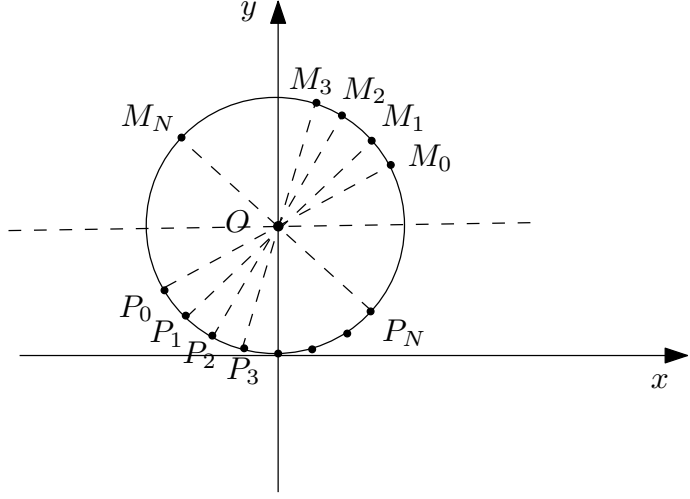}
 \caption{Illustration of the position of points $P_{i}$'s and $M_{i}$'s.}
  \label{fig:point}
\end{figure}

We  construct a sequence of points $P_i$ ($i=0,1,2,...,N$) on the
lower half of circle $C(O,r)$ in counterclockwise order, and the
point $P_{N/2}$ overlaps with point $(0,0)$, see
Figure~\ref{fig:point}. Note that the distance between two adjacent
points $P_{i}$ and $P_{i+1}$ is
$L=2\frac{\epsilon}{2}\sin\frac{\theta}{2}=\epsilon\sin\frac{\theta}{2}$,
and $\angle P_iOP_{i+1}=\theta$, for $i=0,1,2,...,N-1$. Hence, all
the points $P_i,i=0,1,2,...,N$ are within a region of diameter less
than $\frac{\pi}{20}\cdot\frac{\epsilon}{2}<\epsilon/10$. (We
comment that in Figure~\ref{fig:point}, these points are spread out
much more than they should be, as we need the space for putting the
labels.)

We then construct a sequence of points $M_i$ ($i=0,1,2,...,N$) on
the upper half of circle $C(O,r)$ in counterclockwise order. Each
line $\overline{P_iM_i}$ crosses the center of $C(O,r)$; namely
$M_i$ is the symmetry of the point $P_i$ about point $O$.

It is obvious that $d(M_{i},P_i)=\epsilon$ and
$d(M_{i},P_j)<\epsilon, i\neq j$. Recall that $D(P_i,\epsilon')$ is
the neighborhood (disk) centered at $P_i$ with radius be
$\epsilon'$. Here, we have $\epsilon'=\frac{1}{2}
\min\{\epsilon-d(P_i,M_j),i\neq j\}$; moreover,
$d(p,q)\leq(\epsilon-2\epsilon')+2\epsilon'\leq \epsilon$, for $p\in
D(P_i,\epsilon')$, $q\in D(M_j,\epsilon')$, and $i \neq j$.

Let $P'_i$ be the symmetry of the point $P_i$ along the horizontal
line $y=-3\epsilon/4$. Let $M'_i$ be the symmetry of the point $M_i$
along the horizontal line $y=-3\epsilon/4$. We finish the steps of
our construction in order as follows.

\begin{figure}
 \centering
 \includegraphics[width=0.4\textwidth]{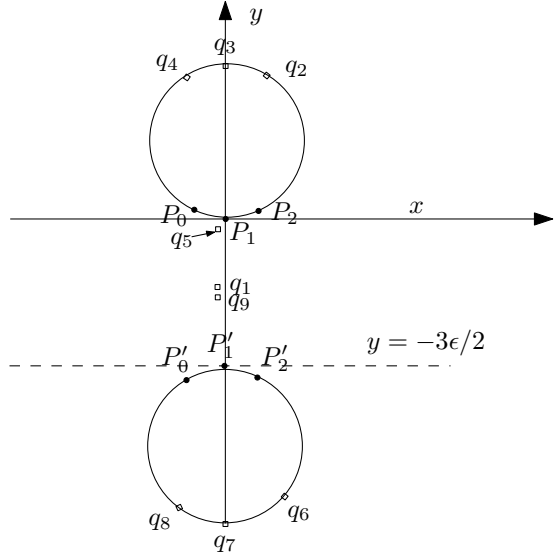}
 \caption{Illustration of the deployment of vertices of $P_i$'s and $Q$, corresponding to Figure~\ref{fig:spread_fff}.}
  \label{fig:compose}
\end{figure}

\begin{enumerate}
\item The imprecise vertices $\{q_2,...,q_{m+1}\}$ used to construct
the clause gadget are deployed in the upper half of the  circle
$C((0,\frac{\epsilon}{2}),\frac{\epsilon}{2})$. An example is given
in Figure~\ref{fig:compose} (see $q_2,q_3,q_4$ there). For an
imprecise vertex $q_i, 2 \leq i \leq m+1$, three points $h_{i_1},
h_{i_2}, h_{i_3} $ are located in $D(P_{i-2},\epsilon')$, $q_i$ is
located in $D(M_{i-2},\epsilon')$, and the three circles
$C(h_{i_1},\epsilon), C(h_{i_2},\epsilon), C(h_{i_3},\epsilon)$
cover $q_i$ as in Figure~\ref{fig:location}. The clause gadget is
constructed as follows. The point $h_{i_2}$ overlaps with $P_{i}$,
$h_{i_1}$ is located to the left of line $\overline{P_{i-2}M_{i-2}}$
with a distance $\epsilon'/3$ to $h_{i_2}$, and $h_{i_3}$ is located
to the right of line $\overline{P_{i-2}M_{i-2}}$ with a distance
$\epsilon'/3$ to $h_{i_2}$. The points $h_{i_1},h_{i_2},h_{i_3}$ are
located on the same line perpendicular to line
$\overline{P_{i-2}M_{i-2}}$. The three intersections between
$C(h_{i_1},\epsilon), C(h_{i_2},\epsilon), C(h_{i_3},\epsilon)$ are
$s_1,s_2,s_3$ from left to right about the horizontal line
$\overline{h_{i_1}h_{i_3}}$. Let $q_i$ be the rectangle with length
$2 \epsilon'/3$ and width $\epsilon''$, the upper long side of $q_i$
crosses $s_1,s_3$ and is symmetric along the line $P_{i-2}M_{i-2}$,
and the lower long side of $q_i$ crosses $s_2$.
$\epsilon''<d(s_2,M_{i-2})=\epsilon-\sqrt{\epsilon^2-(\epsilon'/3)^2}<\epsilon'/3$.
Hence $d(M_{i-2},q)<\epsilon'$ when $q\in q_i$.

The imprecise vertices $\{q_j| 2 \leq j \leq m+1, j \neq i\}$ are
fully covered by $C(h_{i_1},\epsilon)$, $C(h_{i_2},\epsilon)$,
$C(h_{i_3},\epsilon)$ as $d(p,q)<\epsilon$, $p\in D(P_i,\epsilon'),
q\in D(M_j,\epsilon'),i \neq j$. The above design can ensure that,
in the corresponding free space grid graph, either one of the three
cells $C[i,i_1],C[i,i_2], C[i,i_3]$ can be passed by a potential
monotone path, and the cells $C[i,j_1], C[i,j_2], C[i,j_3], j\neq
i,$ can also be passed ($h_{j_1}, h_{j_2}, h_{j_3}$ and  $q_j$, $j
\neq i$, are used to construct another clause gadget), while the
rest of cells in the $i$-th row cannot be passed.

\item The imprecise vertices $\{q_{m+3},...,q_{n+m+2}\}$ used to
construct the variable gadgets are deployed in the lower half of
another circle $C((0,-2\epsilon),\frac{\epsilon}{2})$. For an
example, see Figure~\ref{fig:compose}. For an imprecise vertex $q_i,
m+3 \leq i \leq n+m+2$, two points $h_{i_1}, h_{i_2}$ are located in
$D(P'_{i-m-3},\epsilon')$, $q_i$ is located in
$D(M'_{i-m-3},\epsilon')$ and the long side of  $q_i$ is parallel to
the tangent line at $M'_{i-n-3}$. Two circles $C(h_{i_1},\epsilon),
C(h_{i_2,\epsilon})$ cover $q_i$ as in
Figure~\ref{fig:spread_variable}(a) to construct a variable gadget.
The other imprecise vertices $\{q_j| m+3 \leq j \leq n+m+2, j \neq
i\}$ are fully covered by $C(h_{i_1},\epsilon),
C(h_{i_2},\epsilon)$, the precise location is similar to the
construction in Figure~\ref{fig:location}. This design can ensure
that, in the corresponding free space grid graph, one of the cells
$C(i,i_1),C(i,i_2)$ can be passed by a potential monotone path, and
the cells $C(i,j_1), C(i,j_2),j\neq i$ can also be passed ($h_{j_1},
h_{j_2}$ and $q_j$ are used to construct another variable gadget),
but the other cells in the $i$-th row cannot be passed.

\item The first imprecise vertex $q_1$ and last imprecise vertex
$q_{n+m+3}$ are deployed in $D((0,-3\epsilon/4),\epsilon')$, which
are fully covered by all the circles $C(p,\epsilon), p\in \bigcup
D(P_i,\epsilon')$ or $p\in \bigcup D(P'_i,\epsilon')$, as
$(\frac{\epsilon}{10}+\frac{3\epsilon}{4}+2\epsilon')<\epsilon$.
This design can ensure that, in the corresponding free space grid
graph, all the cells in the first row and last row can be passed by
a potential monotone path.

\item The imprecise vertex $Q_{m+2}$ is deployed  inside region $D((0,0),\epsilon')$, which is only fully
covered by any circle $C(p,\epsilon), p\in \bigcup D(P_i,\epsilon')$
. But it is not covered by the circle $C(p,\epsilon), p\in \bigcup
D(P'_i,\epsilon')$. This design above can ensure that, in the
corresponding free space grid graph, all the cells
$C[m+2,i_1],C[m+2,i_2],C[m+2,i_3]$ in the $m+2$-th row can be passed
($h_{i_1}, h_{i_2}, h_{i_3}$ are used to construct the clause gadget
with $q_i$), while the other cells in the $(m+2)$-th row cannot be
passed.

\item At last, we adjust the order and rename for the vertices in $H$. If
there is a variable cell in the $j$-th row of the free space grid
graph, and there are a total of $k$ clause cells connecting to it,
say the row number of  $k$ cells are $1',2',...,i',...,k'$
respectively, then we choose one point (never be renamed before)
from the three points $h_{i'_1},h_{i'_2},h_{i'_3}$ for each fixed
$i'$, and rename those $k$ points as $h_{j-k}, h_{j-k+1},...,
h_{j-1}$ in order.
\end{enumerate}

\begin{figure}
 \centering
 \includegraphics[width=0.6\textwidth]{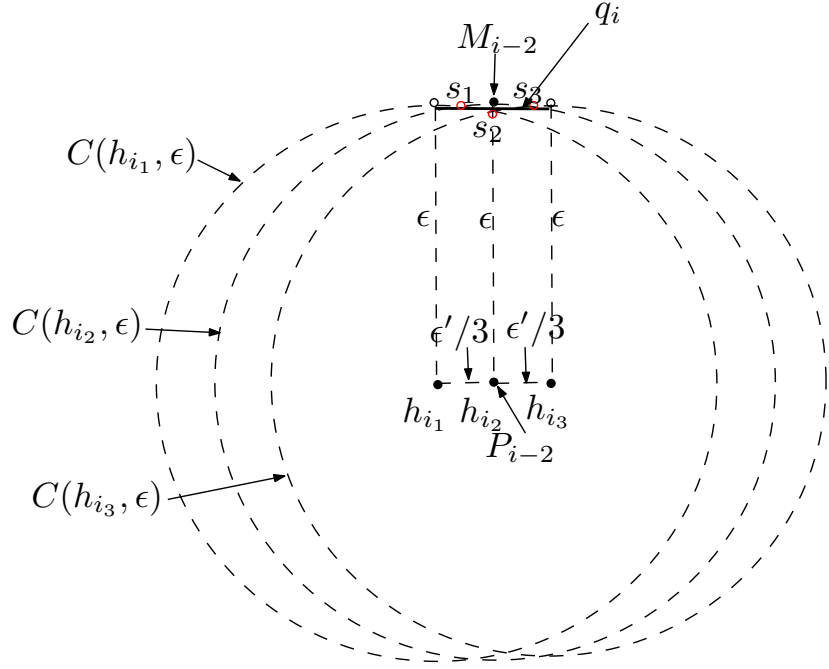}
 \caption{Illustration of the precise location of $h_{i_1},h_{i_2},h_{i_3}$ and $q_i$.}
  \label{fig:location}
\end{figure}

\end{document}